\documentclass[a4paper]{article}

\usepackage{tikz}

\usepackage{amsmath} 
\usepackage{amssymb}
\usepackage{stmaryrd} 

\usepackage{a4wide}

 \usepackage[pdfauthor={Johannes Marti},
             pdftitle={Conditional Logic is Complete for Convexity in the Plane},
            ]{hyperref}

\usepackage{amsthm}
\theoremstyle{plain}
\newtheorem{theorem}{Theorem}[section]

\newtheorem{lemma}[theorem]{Lemma}
\newtheorem{proposition}[theorem]{Proposition}

\newtheorem{problem}[theorem]{Problem}

\newtheorem{maintheorem}[theorem]{Theorem}

\theoremstyle{definition}

\newtheorem{example}[theorem]{Example}
\newtheorem{remark}[theorem]{Remark}

\newcommand{\C}{\mathcal{C}}
\newcommand{\D}{\mathcal{D}}
\newcommand{\F}{\mathcal{F}}
\newcommand{\hull}[1]{\hullSym({#1})}
\newcommand{\hullSym}{\mathrm{co}}
\newcommand{\comp}[1]{\overline{#1}}
\newcommand{\exiimage}[1]{{#1^\exists}}
\newcommand{\extreme}[1]{\mathrm{ex}({#1})}
\newcommand{\ext}[1]{{\llbracket {#1} \rrbracket}}
\newcommand{\Lang}{\mathcal{L}}
\newcommand{\N}{\mathbb{N}}
\newcommand{\Prop}{\mathsf{Prop}}
\newcommand{\powerset}{\mathcal{P}}
\newcommand{\X}{\mathcal{X}}
\newcommand{\R}{\mathbb{R}}
\newcommand{\uniimage}[1]{{#1^\forall}}
\newcommand{\upset}[1]{{{#1}\!\!\uparrow}}
\newcommand{\upseti}[2]{{{#2}\!\!\uparrow_{#1}}}
\newcommand{\upsets}[1]{\mathcal{U}({#1})}

\renewcommand{\iff}{\quad \mbox{iff} \quad}

\def\mathrlap{\mathpalette\mathrlapinternal}%
\def\mathrlapinternal#1#2{\rlap{$\mathsurround=0pt#1{#2}$}}%
\makeatletter
\newdimen\@mydimen%
\newdimen\@myHeightOfBar%
\newcommand{\SetScaleFactor}[1]{%
    \settoheight{\@mydimen}{#1}%
    \pgfmathsetmacro{\scaleFactor}{\@mydimen/\@myHeightOfBar}%
}

\newcommand*{\Scale}[2][3]{\scalebox{#1}{\ensuremath{#2}}}%
\newcommand{\quickDefaults}{\mathrel{\Scale[0.75]{|\mathrlap{\kern-0.48ex\sim}\hphantom{\kern-0.41ex\sim}}}}

\newcommand{\cond}{\leadsto}

\author{Johannes Marti}

\title{Conditional Logic is Complete for Convexity in the Plane} 

\begin{document}

\maketitle

  \begin{abstract}
\noindent
We prove completeness of preferential conditional logic with respect to convexity over finite sets of points in the Euclidean plane. A conditional is defined to be true in a finite set of points if all extreme points of the set interpreting the antecedent satisfy the consequent. Equivalently, a conditional is true if the antecedent is contained in the convex hull of the points that satisfy both the antecedent and consequent. Our result is then that every consistent formula without nested conditionals is satisfiable in a model based on a finite set of points in the plane. The proof relies on a result by Richter and Rogers showing that every finite abstract convex geometry can be represented by convex polygons in the plane.

  \end{abstract}

\section{Introduction}

Preferential conditional logic was introduced by Burgess
\cite{Burgess81} and Veltman \cite{Veltman85} to axiomatize the
validities of the conditional with respect to a semantics in models
based on ordering relations. In this semantics a conditional $\varphi
\cond \psi$ is true with respect to an order over a finite set of worlds
if the consequent $\psi$ is true at all worlds that are minimal in the
order among the worlds at which the antecedent $\varphi$ is true.
Preferential conditional logic is sound and complete in this semantics
with respect to models that are based on arbitrary preorders. But both
Burgess and Veltman note that for completeness it suffices to consider
partial orders. The axioms of preferential conditional logic are a
weakening of the axioms in Lewis' conditional logic \cite{Lewis73} that
is sound and complete for models that are based on strict weak orders,
which are in bijective correspondence with total preorders.


Similar semantic clauses as in conditional logic, and thus analogous
axiomatic systems, have later also been used in default reasoning
\cite{Shoham88,Kraus90}, in belief revision theory \cite{Grove88,Rott09}
and in dynamic epistemic logic \cite{BaltagSmets06,Benthem07b}. It
should also be mentioned that the axiomatizations of conditional logics
with respect to their order semantics are similar to the
characterizations of choice functions that are rationalizable by some
preference relation \cite{Arrow59,Sen71}. Moreover, the semantic clause
for the conditional in orders, which is often attributed to
\cite{Lewis73}, goes back to an earlier semantic clause for conditional
obligations in deontic logic \cite{Hansson69}.

Preferential conditional logic has also been shown to be complete with
respect to semantic interpretations that are quite different from the
semantics in terms of partial orders. Most notable are the
interpretation of validity of inferences between conditionals as
preservation of high conditional probability \cite{Adams75,Geffner92}
and premise semantics, where the conditional is interpreted relative to
a premise set. A premise set is a family of sets of worlds, thought of
as propositions that encode relevant background information from the
linguistic context \cite{Veltman76,Kratzer81}. In this paper we provide
yet another interpretation to preferential conditional logic. We show
that it is complete with respect to convexity over finite sets of points
in the Euclidean plane. This places conditional logic into the tradition
of modal logics with a natural spatial or geometric semantics
\cite{Benthem07}, most famous of which is the completeness result for S4
with respect to the topology of the real line by McKinsey and Tarski
\cite{McKinsey44,Bezhanishvili05}. 

\begin{figure}
\begin{center}
\begin{tikzpicture}[scale=1]
  \node (innerdown) at (8,3.4) {
$\begin{array}{l}
 \mbox{true conditionals:} \\
 (p \lor q) \cond r \\
 (\neg p \lor \neg q) \cond \neg p \\
 \top \cond (q \leftrightarrow r) \\
 \\
 \mbox{false conditionals:} \\
 p \cond r \\
 \neg r \cond \neg q \\
 \top \cond r
 \end{array}$
};
  \filldraw[dotted, color=black!80, fill=black!8] (4,5) -- (0,5) -- (2.4,3);
  \draw[dotted, color=black!80] (4,5) -- (1.2,1.5);
 \node (x) at (0,5) {$p q r$};
 \node (y) at (4,5) {$\bar{p} q r$};
 \node (z) at (2.4,3) {$p \bar{q} r$};
 \node (u) at (1.9,4.3) {$p q \bar{r}$};
 \node (v) at (1.2,1.5) {$\bar{p} \bar{q} \bar{r}$};
\end{tikzpicture}
\end{center}
\caption{A finite set of points in the plane and examples of
conditionals that are true or false relative to this set of points.}
\label{intro example}
\end{figure}

To illustrate our semantics consider the finite set of points in
Figure~\ref{intro example}. Think of these points as satisfying
propositional letters as indicated in their label. For instance the
point $\bar{p}qr$ in the upper right corner satisfies $q$ and $r$ but
not $p$. Our semantics is such that a conditional $\varphi \cond \psi$
is true relative to such a set of points if the set of points at which
$\varphi$ is true is completely contained in the convex hull of the set
of points at which both $\varphi$ and $\psi$ are true. Recall that a
convex set is a set that for any two points in the set also contains the
complete line segment between these points. Intuitively, these are the
sets without holes or dents. The convex hull of a set is the least
convex set that contains the set. As an example of a convex hull we have
in Figure~\ref{intro example} that the shaded area is the convex hull of
the three points $pqr$, $\bar{p}qr$ and $p\bar{q}r$. In this example
the conditional $(p \lor q) \cond r$ is true because all points at which
$p \lor q$ is true are contained in the convex hull of the points where
$p \lor q$ and $r$ are both true. The conditional $p \cond r$ is however
not true in the example because the point $pq\bar{r}$ satisfies $p$ but
is not contained in the convex hull of the points $pqr$ and
$p\bar{q}r$, which are all the points that satisfy $p$ and $r$.

An equivalent formulation of our semantic clause is that a conditional
$\varphi \cond \psi$ is true if the consequent $\psi$ is true at all the
extreme points of the set of points where the antecedent $\varphi$ is
true. An extreme point of some set is a point in the set that is not in
the convex hull of all the other points from the set. Intuitively, the
extreme points of some set are the outermost points of that set. In the
example from Figure~\ref{intro example} we have that $pqr$, $\bar{p}qr$
and $p\bar{q}r$ are the extreme points of the set that is shaded. On the
other hand $pq \bar{r}$ is not an extreme point of the shaded set
because it is in the convex hull of the points $pqr$, $\bar{p}qr$ and
$p\bar{q}r$. Note that in this formulation of the semantic clause for a
conditional $\varphi \cond \psi$ the extreme points of the set of points
satisfying the antecedent $\varphi$ play a role that is analogous to the
minimal $\varphi$-worlds in the order semantics. Conversely, we will see
later that the upwards closed set in an order play a role that is
analogous to the convex sets in the geometric semantics.

In this paper we focus on a semantics that is only defined for formulas
that do not contain nested conditionals and in which all propositional
letters occur in the scope of a conditional. It is possible to overcome
this restriction but this has no significant influence on the axiomatic
questions that we are concerned with.


The main result of our paper can be formulated as follows: All finite
constellation of points in the plane of the kind as shown in
Figure~\ref{intro example} satisfy all the theorems in preferential
conditional logic and every formula that is not a theorem of the logic
is false in some such constellation.

The completeness proof from this paper consists of two steps:
\begin{enumerate}
 \item We first observe that preferential conditional logic is complete
for a semantics in models based on convex geometries.
 \item We then show that every finite convex geometry can be represented
by a finite set of points in the plane in such a way that all true
formulas of conditional logic are preserved.
\end{enumerate}
From these two steps we obtain our completeness result because by the
first step every consistent formula $\varphi$ is true in some finite
model based on a convex geometries and by the second step this model can
be transformed into a concrete model of $\varphi$ that is based on a
finite set of points in the plane. We now describe these two steps in
greater detail. 

In the first step we make use of the notion of a convex geometry
\cite{Edelman85,Korte91,Adaricheva16}. Formally, convex geometries are
families of sets that are closed under arbitrary intersections and have
the anti-exchange property, which is a separation property that is
reminiscent of the $T_0$ separation property in topology. Convex
geometries are a combinatorial abstraction of the notion of a convex set
in Euclidean spaces, such as the Euclidean plane. This is somewhat
analogous to how topological spaces are an abstraction from the notions
of open and closed sets in Euclidean spaces. The convex sets in any
subspace of an Euclidean space form a convex geometry. But it is not the
case that every convex geometry, or even every finite 
convex geometry, is isomorphic to a subspace of some Euclidean space. An
easy way to see this is to observe that in any Euclidean space all
singleton sets are convex, which is not enforced by the definition of a
convex geometry.

One can view the semantics in convex geometries as a generalization of
the order semantics over partial orders. The family of upwards closed
sets in any partial order form a convex geometry. Moreover, a
conditional is true relative to a given partial order if and only if it
is also true in the convex geometry of all upwards closed sets in the
order. Note that this especially means that the completeness of the
order semantics entails the completeness of the semantics in convex
geometries.

To understand the relation between the order semantics, the semantics in
convex geometries and the semantics for convexity between
finitely many points in the plane it might be helpful to think of an
analogy with the different semantics for the modal logic S4. Both,
preferential conditional logic and S4, have a relatively concrete
relational semantics in terms of partial orders for preferential
conditional logic and in terms of preorders, that are transitive and
reflexive relations, for S4. Both logics have an abstract spatial or
geometric semantics, the semantics in convex geometries for preferential
conditional logic and the semantics in topological spaces for S4. In
both cases the abstract semantics generalizes the relational semantics.
For preferential conditional logic this is done by considering the
upward closed sets in the partial order as a convex geometry. For S4 one
can also considers the upwards closed sets in a preorder, which form a
so called Alexandroff topology. Both logics additionally have a concrete
spatial or geometric semantics, over a finite set of points in the plane
for preferential conditional logic, and over the whole real line for S4.
In both cases proving completeness for the concrete spatial or geometric
semantics requires extra work. For preferential conditional logic this
is the construction mentioned in the second step above and in the case
of S4 it is the theorem of McKinsey and Tarski.

The semantics in convex geometries can also be seen as a
further development of premise semantics. The convex sets in our
semantics play the role of the complements of the sets of worlds in the
premise set of premise semantics. There is, however, a crucial
difference in the semantic clause with which a conditional is
interpreted in a family of sets of worlds. Motivated by linguistic
considerations premise semantics uses a quite sophisticated semantic
clause that is insensitive to closing the family of sets under
intersections. In \cite{Negri15,Girlando21} it is observed that for
developing proof systems for preferential conditional logic it is
beneficial to lift the implicit assumption that the family of sets of
worlds, relative to which the conditional is evaluated, is closed under
intersections. To achieve this they use a simplified semantic clause
from \cite{MartiPinosio14a} that is sensitive to closure under
intersections. When one uses the conditional with this semantic clause
relative to a family of sets of worlds that is not closed under
intersection different formulas turn out to be true than would be true
relative to the same family of sets of worlds using the semantic clause
from premise semantics. Hence, it is helpful to distinguish this new
setting from premise semantics and call it neighborhood semantics,
following the terminology form \cite{Negri15}.

This neighborhood semantics is also the starting point for the
categorical correspondence in \cite{MartiPinosio19}. Based on earlier
work on the theory of choice functions \cite{Koshevoy99,Johnson01} this
paper establishes a correspondence between finite Boolean algebras with
additional structure that encodes non-nested preferential conditional
logic and families of subsets of the set of atoms of these algebras. To
obtain a well-behaved correspondence it is necessary to allow for
families of sets that are not closed under intersections. However, one
can require closure under unions and a separation property that is dual
to the anti-exchange property mentioned above. If one then considers the
complements of all the sets in a such a family of sets then one obtains
a new family that is closed under arbitrary intersections and that has
the anti-exchange property. Thus, one gets a convex geometry.

The second step of the proof is to show that for every finite convex
geometry there is a finite subspace of the plane that satisfies the
same formulas in conditional logic. This step is not trivial because, as
we already explained above, not every finite convex geometry is
isomorphic to a subspace of some Euclidean space. However, following
\cite{Kashiwabara05}, there has recently been a lot of literature on
representing finite convex geometries inside of Euclidean spaces by some
more intricate construction than just selecting an isomorphic subspace
\cite{Czedli14,Czedli16,Richter17,Adaricheva19}. The main result of
\cite{Kashiwabara05}, for which \cite{Richter17} give a much shorter
proof, is that every finite convex geometry is isomorphic to the convex
geometry on a finite set of points in some Euclidean space, if we use an
alternative notion of convex set that is slightly different from the
standard notion of convex set. Moreover, \cite{Richter17} show that
every finite convex geometry is isomorphic to the convexity over a set
of polygons in the plane, using the standard notion of convexity, but
now every point in the original convex geometry corresponds to a whole
polygon in the plane. The papers \cite{Czedli14,Czedli16,Adaricheva19}
investigate to what extent it is possible to prove the same result using
circles instead of polygons.

In the second step of the completeness proof we make use of the
representation by \cite{Richter17}, where a finite convex geometry is
represented by a set of polygons. This construction is such that the
extreme points of any two polygons in the set are disjoint. One can thus
define a function that maps an extreme point of some polygon in the set
to the point in the original convex geometry that the polygon is
representing. The domain of this function can be considered to be the
finite subspace of the plane consisting of all the points that are an
extreme point of one of the polygons. The crucial insight is then that
this function is a strong morphism of convex geometries in a sense
defined in \cite{MartiPinosio19}, which guarantees the preservation of
true formulas in conditional logic.

The structure of this paper is as follows: In Section~\ref{convex
geometries} we review the notion of a convex geometry and fix the
necessary terminology. In Section~\ref{conditional} we present the
syntax of preferential conditional logic and define its semantics in
convex geometries. Section~\ref{completeness} contains a self contained
completeness result for preferential conditional logic with respect to
its semantics in convex geometries. In Section~\ref{morphisms} we
discuss the notion of morphism between finite convex geometries from
\cite{MartiPinosio19} that preserves the truth of all formulas in
conditional logic. In Section~\ref{representation} we show that the
representation of finite convex geometries in the plane from
\cite{Richter17} yields such a morphism. In Section~\ref{euclid} we put
the results from the previous sections together to prove the
completeness of preferential conditional logic with respect to convexity
between finite sets of points in the plane. Moreover, we show that this
result can not be improved to a completeness results with respect to
sets of points on the real line.

\section{Convex geometries}
\label{convex geometries}

In this section we recall some basic terminology and results related to
abstract convex geometries. For a more thorough introduction see
\cite{Edelman85,Adaricheva16} or \cite[ch.~3]{Korte91}

\subsection{Basic definitions}

A \emph{convex geometry} $(W,\C)$ is a set $W$ together with a family
$\C \subseteq \powerset W$ of \emph{convex sets} that has the following
properties:
\begin{enumerate}
 \item $\C$ is closed under arbitrary intersections, that is, $\bigcap
\X \in \C$ for all $\X \subseteq \C$.
 \item $\C$ has the \emph{anti-exchange property} that for every $C \in
\C$ and all $x,y \in W$ with $x,y \notin C$ and $x \neq y$ there is a $D
\in \C$ with $C \subseteq D$ such that $x \in D$ and $y \notin D$, or $x
\notin D$ and $y \in D$.
\end{enumerate}
We sometimes use just $W$, or just $\C$, to denote the convex geometry
$(W,\C)$ consisting of both $W$ and $\C$. Thereby it is assumed that the
identity of the other component is understood from the context.

Most authors require that $\emptyset \in \C$. We do not require this
because, as we explain in Remark~\ref{impossible worlds}, it is
convenient for the semantics of conditional logic to allow for convex
geometries in which the empty set is not convex.

We call the complements of convex sets \emph{feasible}, following the
literature on antimatroids \cite[ch.~3]{Korte91}. The family of all
feasible sets is denoted by $\F = \{W \setminus C \mid C \in \C\}$. We
use the notation $\comp{X} = W \setminus X$ to denote the complement of
some $X \subseteq W$.

Given any subset $X \subseteq W$ its \emph{convex hull} $\hull{X}
\subseteq W$ is defined as
\[
 \hull{X} = \bigcap \{C \in \C \mid X \subseteq C\}.
\]
Because convex sets are closed under intersection the convex hull
$\hull{X}$ is a convex set. In fact it is the least convex set
containing $X$. One can also show that as an operation on $\powerset W$
the convex hull $\hullSym : \powerset W \to \powerset W$ defines a
closure operator, meaning that $X \subseteq Y$ implies $\hull{X}
\subseteq \hull{Y}$, $X \subseteq \hull{X}$, and $\hull{\hull{X}}
\subseteq \hull{X}$ for all $X,Y \subseteq W$. The relation between the
family of convex sets and the convex hull is an instance of the
well-known correspondence between complete meet-semilattices and closure
operators.


For every subset $X \subseteq W$, where $(W,\C)$ is a convex geometry,
we define the \emph{relative convexity} on $X$ as follows: A set $C
\subseteq X$ is convex in the relative convexity if there is some set
$C'$ that is convex in $W$ such that $C = C' \cap X$. It is not hard to
see that the relative convexity is a convex geometry.

The prime example of convex geometries are the families of convex sets
in the Euclidean space $\R^n$ for every dimension $n$. A set $C
\subseteq \R^n$ is convex if it contains the complete line segment
between any two of its points. This means that for all $x,y \in C$ we
need $\{\lambda x + (1 - \lambda)y \mid \lambda \in [0,1]\} \subseteq
C$. We call the family of convex sets defined in this way the
\emph{standard convexity}. It is well know that the convex hull of a set
$X \subseteq \R^n$ in the standard convexity is the set of all convex
combinations of points in $X$, where a \emph{convex combination} of
$x_1,\dots,x_k \in X$ is any point that can be written as $\sum_{i =
1}^k \lambda_i x_i$, for $\lambda_1,\dots,\lambda_k \geq 0$ with
$\sum_{i = 1}^k \lambda_i = 1$.

Another example of convex geometries are partially ordered sets. Because
the standard semantics of conditional logic is usually defined over
partially ordered sets this example provides the link between convex
geometries and conditional logic. Recall that a partially ordered set,
or just \emph{poset}, is a set $W$ together with a partial order $\leq$
on $W$, where a partial order is a binary relation that is reflexive,
transitive and anti-symmetric. Given a partial order $\leq$ on $W$ we
define the \emph{upset convexity} $\upsets{\leq}$ on $W$ to consists of
all the \emph{upward closed sets} in $\leq$, that is, all the sets $C$
such that $x \in C$ and $x \leq y$ implies $y \in C$. The convex hull of
a set $X \subseteq W$ is then identical to the set $\upset{X} = \{y \in
W \mid x \leq y \mbox{ for some } x \in X\}$, which is the \emph{upwards
closure} of $X$. Note that $\upsets{\leq}$ is just the Alexandroff
topology associated to the order $\leq$. Closure under arbitrary
intersections is thus obvious. The anti-exchange property follows from
the $T_0$ separation property of any Alexandroff topology that is
defined from a poset. The reason that in this paper we assume that the
order semantics of conditional logic is based on posets instead of just
preorders is that the Alexandroff topology of a preorder that is not
anti-symmetric does not have the $T_0$ separation property and thus it
is not a convex geometry.

\subsection{Extreme points}

A point $x \in X$ in some set $X \subseteq W$ in a convex geometry
$(W,\C)$ is an \emph{extreme point} of $X$ if $x \notin \hull{X
\setminus \{x\}}$. The intuition is that an extreme point of $X$ is an
outermost point of the set $X$. The extreme points of a set in the upset
convexity of a poset are precisely the minimal elements of the set. We
write $\extreme{X} \subseteq X$ for the set of all of its extreme points
of $X$. The following proposition yields an alternative characterization
for the set of extreme points.
\begin{proposition} \label{extreme reformulation}
 For every $X \subseteq W$ we have $\extreme{X} = \bigcap \{Y \subseteq
X \mid X \subseteq \hull{Y}\}$.
\end{proposition}
\begin{proof}
 For the contrapositive of the $\subseteq$-inclusion take $x \in X$ such
that there is some $Y \subseteq X$ with $x \notin Y$ and $X \subseteq
\hull{Y}$. Then $x \in \hull{Y} \subseteq \hull{X \setminus \{x\}}$ and
so $x$ is not an extreme point of $X$.

 For the contrapositive of the $\supseteq$-inclusion consider an $x \in
X$ with $x \in \hull{X \setminus \{x\}}$. Set $Y = X \setminus \{x\}$,
and observe that $X \subseteq \hull{Y}$ but $x \notin Y$.
\end{proof}

For finite sets one has the following relation between extreme points
and the convex hull operator.
\begin{theorem} \label{smooth theorem}
The following are equivalent for every finite set $K \subseteq W$ in a convex
geometry $\C$ on $W$:
\[
 \extreme{K} \subseteq X \iff K \subseteq \hull{K \cap X} \qquad \mbox{
for all } X \subseteq W.
\]
\end{theorem}
\begin{proof}
 This follows from item~3 of Theorem~1 in \cite{MartiPinosio19} and the
observation that every finite set is smooth in the terminology of that
paper. Note that the proof of this theorem uses the characterization
from Proposition~\ref{extreme reformulation} as the definition of the
extreme points.
\end{proof}

Lastly, we define the notion of a polygon. A \emph{polygon} $P \subseteq
W$ in a convex geometry $(W,\C)$ is any set that can be written of the
form $P = \hull{P'}$ for a finite set $P' \subseteq P$. Clearly every
such polygon has only a finite number of extreme points because for any
$x \in P$ with $x \notin P'$ we have that $x \in \hull{P'} \subseteq
\hull{P \setminus \{x\}}$.

\section{Conditional logic}
\label{conditional}

In this section we discuss the syntax of preferential conditional logic
that we use in this paper and explain its semantics in convex
geometries.

\subsection{Syntax of one-step preferential conditional logic}

Conditional logics are commonly formulated in a classical propositional
modal language with one binary modality $\cond$, which forms the
\emph{conditional} $\varphi \cond \psi$ with \emph{antecedent} $\varphi$
and \emph{consequent} $\psi$ \cite{Lewis73,Chellas75}. That $\cond$ is a
modality means that one can nest conditionals, as for example in the
formula $(((p \cond q) \cond r) \land q) \rightarrow r$. In this paper
we choose not to deal with the complications arising from nested
conditionals and instead just work with one-step formulas that are
Boolean combination of conditionals over propositional formulas. This is
not a substantial restriction for most conditional logics, because the
axiomatizations of these logics constrain only one layer of conditionals
and then are extended freely to formulas of larger conditional depth.
Readers familiar with coalgebraic modal logic might recognize this as
the one-step setup that is common in coalgebraic logic \cite{Kupke11}.
We sketch in Remarks \ref{not one-step}~and~\ref{nested completeness}
below how one would extend our semantics and completeness result to
formulas with nested conditionals.

To be more precise about our setting fix a set $\Prop$ of propositional
letters and consider the grammar
\begin{align*}
 \varphi_0 & ::= p \mid \neg \varphi_0 \mid \varphi_0 \land \varphi_0, &
\mbox{where } p \in \Prop, \\
 \varphi_1 & ::= \varphi_0 \cond \varphi_0 \mid \neg \varphi_1 \mid
\varphi_1 \land \varphi_1.
\end{align*}
Let $\Lang_0$ be the set of formulas generated from $\varphi_0$ and
$\Lang_1$ the set of formulas generated from $\varphi_1$. Note that
$\Lang_0$ is just the language of classical propositional logic. In both
$\Lang_0$ and $\Lang_1$ we use further Boolean connectives, such as
$\lor$, $\rightarrow$, and $\leftrightarrow$, as abbreviations with their
usual meaning in classical logic. To omit parenthesis we assume that
$\neg$ binds stronger than $\land$ and $\lor$, which in turn bind
stronger than $\cond$, $\rightarrow$ and $\leftrightarrow$.

In our axiomatization of preferential conditional logic we follow the
one-step setup in that we only consider proofs in which all formulas are
either from $\Lang_0$ or from $\Lang_1$. Hence, proofs are not
allowed to contain nested conditionals or formulas with conditionals
that contain propositional letters that are not in the scope of a
conditional.

As axioms we allow all instances of propositional tautologies in
$\Lang_0$ plus the following axioms that are in $\Lang_1$:
\begin{align*}
 \mbox{(Id)} \ & p \cond p, &
 \mbox{(And)} \ & (p \cond q) \land (p \cond r) \rightarrow (p \cond q
\land r), \\
 \mbox{(CM)} \ & (p \cond q) \land (p \cond r) \rightarrow (p \land r
\cond q), &
 \mbox{(Or)} \ & (p \cond q) \land (r \cond q) \rightarrow (p \lor r
\cond q).
\end{align*}
We have the following inference rules: First, modus ponens, where the
premises are either both in $\Lang_0$ or both in $\Lang_1$; second,
uniform substitution $\varphi/\varphi[\sigma]$, where either $\varphi
\in \Lang_0$ and $\sigma : \Prop \to \Lang_i$ for some $i \in \{0,1\}$,
or $\varphi \in \Lang_1$ and $\sigma : \Prop \to \Lang_0$; and third, we
have the following two inference rules, with premises in $\Lang_0$ and
conclusions in $\Lang_1$:
\[
  \mbox{(LLE)} \ \frac{\varphi \leftrightarrow \chi}{(\varphi \cond \psi)
\leftrightarrow (\chi \cond \psi)}, \qquad \mbox{and} \qquad
 \mbox{(RW)} \ \frac{\psi \rightarrow \chi}{(\varphi \cond \psi) \rightarrow
(\varphi \cond \chi)}.
\]
As is common in Hilbert-style axiomatizations we understand these rules
such that the conclusion is derivable whenever the premises are
derivable. In the derivation system given here there is no notion of a
proof with open assumptions, and the rules (LLE) and (RW) would no
longer be sound for proofs with open assumptions.

We use the standard notions of derivability and consistency for formulas
in either $\Lang_0$ or $\Lang_1$ with respect to the above axiomatic
system. We also write $\vdash \varphi$ if some $\varphi \in \Lang_i$ for
some $i \in \{0,1\}$ is derivable

The axioms and rules given here and their names closely follow the rules
of System~P in the literature on nonmonotonic consequence relations
\cite{Kraus90}. It is however easy to show that these rules and axioms
are interderivable with the rules and axioms from \cite{Burgess81} or
\cite{Veltman85}.

The following proposition gathers examples of derivable formulas and
rules.
\begin{proposition} \label{derivable}
 The following formulas are derivable in preferential conditional logic:
\begin{align*}
 \mbox{\normalfont (WCM)} \ & (p \cond q \land r) \rightarrow (p \land q
\cond r), &
 \mbox{\normalfont (CCut)} \ & (p \cond q) \land (p \land q \cond r)
\rightarrow (p \cond  r), \\
 \mbox{\normalfont (S)} \ & (p \land q \cond r) \rightarrow (p \cond
\neg q \lor r), &
 \mbox{\normalfont (CCut')} \ & (p \cond q) \land (q \cond r)
\rightarrow (p \lor q \cond r). \\
\end{align*}
 The following rule is derivable in preferential conditional logic:
\[
  \mbox{\normalfont (R)} \ \frac{\psi \rightarrow \chi}{(\varphi \cond
\psi) \rightarrow ((\varphi \land \chi) \lor \psi  \cond \psi)}.
\]
\end{proposition}
\begin{proof}
 Derivation of (WCM): With (RW) we obtain that $(p \cond q \land r)
\rightarrow (p \cond q)$ and $(p \cond q \land r) \rightarrow (p \cond
r)$ are derivable. Because by (CM) the formula $(p \cond q) \land (p
\cond r) \rightarrow (p \land r \cond q)$ is an axiom we can then use
propositional reasoning to derive $(p \cond q \land r) \rightarrow (p
\land q \cond r)$.

 In the remaining derivations we omit the steps that are propositional
and focus on the axioms or rules involving the conditional. We are
confident that the reader is able to supply the missing details. As an
example a short description of the above derivation of (WCM) would be as
follows: From $p \cond q \land r$ we can derive with the help of (RW)
that $p \cond q$ and that $p \cond r$. With (CM) it follows that $p
\land q \cond q$.

 Derivation of (S): First observe that from (Id) we get that $p \land
\neg q \cond p \land \neg q$ and with (RW) we obtain $p \land \neg q
\cond \neg q \lor r$. Then use (RW) again to obtain $p \land q \cond
\neg q \lor r$ from $p \land q \cond r$. We can use (Or) to get $(p \land q)
\lor (p \land \neg q) \cond \neg q \lor r$. By (LLE) we obtain $p \cond
\neg q \lor r$.

 Derivation of (CCut): From $p \land q \cond r$ it follows by (S) that
$p \cond \neg q \lor r$. Combining this with the assumption $p \cond q$
using (And) we obtain $p \cond (\neg q \lor r) \land q$. By (RW) follows
that $p \cond r$ because $(\neg q \lor r) \land q \rightarrow r$ is a
theorem of classical propositional logic.

 Derivation of (CCut'): First derive $p \lor q \cond q$ using (Or), (Id)
and the assumption $p \cond q$. Then observe that by (LLE) we obtain $(p
\lor q) \land q \cond r$ from the assumption $q \cond r$. Then apply
(CCut) to $p \lor q \cond q$ and $(p \lor q) \land q \cond r$,
substituting the letter $p$ in (CCut) with $p \lor q$. This yields $p
\lor q \cond r$.

 Derivation of (R): Because of the premise that $\psi \rightarrow \chi$
we obtain $\psi \cond \chi$ because of (RW) and the instance $\psi \cond
\psi$ of (Id). Applying (CM) to $\varphi \cond \psi$ and $\psi \cond
\chi$ yields $\varphi \land \chi \cond \psi$. Because $\psi \cond \psi$
holds by (Id) we can use (Or) to get $(\varphi \land \chi) \lor \psi
\cond \psi$.
\end{proof}

\subsection{Semantics of the conditional in convex geometries}
\label{abstract semantics}

To give a semantics to the conditional we are using models that are
based on abstract convex geometries as defined in Section~\ref{convex
geometries}. Thus, we define a \emph{model} $M = (W,\C,V)$ to consist of
\begin{itemize}
 \item a set $W$, whose elements are called \emph{points} or \emph{worlds},
 \item a convex geometry $\C \subseteq \powerset W$ over $W$, and
 \item a function $V : \Prop \to \powerset W$ that is called the
\emph{valuation function}.
\end{itemize}
As is usual in modal logics the valuation function $V$ is used to
assign meanings to the propositional letters in $\Prop$. This assignment
of meanings is extended to propositional formulas from $\Lang_0$ in the
standard way with the recursive clauses
\[
 \ext{p}_V = V(p), \qquad \ext{\neg \varphi}_V = W \setminus
\ext{\varphi}_V, \qquad \mbox{and} \qquad \ext{\varphi \land \psi}_V =
\ext{\varphi}_V \cap \ext{\psi}_V.
\]
We often write $\ext{\varphi}$ for $\ext{\varphi}_V$ if $V$ is clear
from the context.

We use the standard clauses for the propositional connectives over
$\Lang_1$ relative to the model $M = (W,\C,V)$:
\[
 M \models \neg \varphi \iff \mbox{not } M \models \varphi, \qquad
\mbox{and} \qquad M \models \varphi \land \psi \iff M \models \varphi
\mbox{ and } M \models \psi.
\]
The conditional has the following semantics:
\begin{tabbing}
 $M \models \varphi \cond \psi$ \quad iff \quad \= for all $C \in \C$
with $\ext{\varphi} \nsubseteq C$ there is a $D \in \C$ \\
 \> with $C \subseteq D$ and $\ext{\varphi} \nsubseteq D$ such that
$\ext{\varphi} \subseteq D \cup \ext{\psi}$.
\end{tabbing}
The truth of formulas in $\Lang_1$ is only relative to the model $M$ and
does not need to be relativized to a world of evaluation. This is
possible because we do not nest conditionals and all propositional
letters that occur in a formula from $\Lang_1$ need to be in the scope
of some conditional.

We use the standard notion of validity, calling a formula $\varphi \in
\Lang_1$ \emph{valid} iff $M \models \varphi$ for all models $M$. As
usual in modal logic we also call a formula valid over a class of models
or convex geometries if it is true in all models from this class or it
is true in all models that are based on a convex geometry from the
class.

Preferential conditional logic is sound for this semantics. Note that
the proof of soundness never uses the special properties of the convex
geometry $\C \subseteq \powerset W$. Soundness already holds for
arbitrary families of sets.
\begin{proposition}
 If $\varphi \in \Lang_1$ is derivable in preferential conditional logic
then $\varphi$ is valid.
\end{proposition}
\begin{proof}
 One first shows, analogous to the soundness of propositional logic,
that if $\vdash \varphi$ for some $\varphi \in \Lang_0$ then
$\ext{\varphi}_V = W$ for all valuations $V : \Prop \to \powerset W$.
Using this one can show the statement of the proposition with a routine
induction on the length of derivations in the axiomatic system. Here we
only treat the case of the axiom (Or) and leave all other cases, which
are easier, to the reader.

Consider any model $M = (W,\C,V)$. We want to show that $M \models (p
\cond q) \land (r \cond q) \rightarrow (p \lor r \cond q)$. Assume that
$M \models p \cond q$ and that $M \models r \cond q$. To show $M \models
p \lor r \cond q$ consider any convex $C \in \C$ such that $\ext{p \lor
r}_V \nsubseteq C$. We need to find a convex $D \in \C$ with $C
\subseteq D$, $\ext{p \lor r}_V \nsubseteq D$ and $\ext{p \lor r}_V
\subseteq D \cup \ext{q}_V$. Because $\ext{p \lor r}_V \nsubseteq C$ it
follows that either $\ext{p}_V \nsubseteq C$ or $\ext{r}_V \nsubseteq
C$. Consider the case where $\ext{p}_V \nsubseteq C$. The reasoning in
the other case where $\ext{r}_V \nsubseteq C$ is completely analogous.
Because $M \models p \cond q$ it follows from $\ext{p}_V \nsubseteq C$
that there is some $C' \in \C$ with $C \subseteq C'$, $\ext{p}_V
\nsubseteq C'$ and $\ext{p}_V \subseteq C' \cup \ext{q}_V$. Then
distinguish cases depending on whether $\ext{r}_V \subseteq C'$.

If $\ext{r}_V \subseteq C'$ then we can let $D = C'$ because $\ext{p
\lor r}_V \nsubseteq C'$ follows from $\ext{p}_V \nsubseteq C'$ and
$\ext{p \lor r}_V \subseteq C' \cup \ext{q}_V$ follows from $\ext{p}_V
\subseteq C' \cup \ext{q}_V$ together with $\ext{r}_v \subseteq C'$.

If $\ext{r}_V \nsubseteq C'$ then we can apply the assumption $M \models
r \cond q$ to obtain a $C'' \in \C$ with $C' \subseteq C''$, $\ext{r}_V
\nsubseteq C''$ and $\ext{r}_V \subseteq C'' \cup \ext{q}_V$. We can let
$D = C''$. It clearly holds that $C \subseteq C''$. That $\ext{p \lor
r}_V \nsubseteq C''$ follows from $\ext{r}_V \nsubseteq C''$. Lastly, it
holds that $\ext{p \lor r}_V \subseteq C'' \cup \ext{q}_V$ because
$\ext{p}_V \subseteq C' \cup \ext{q}_V$, $C' \subseteq C''$ and
$\ext{r}_V \subseteq C'' \cup \ext{q}_V$
\end{proof}

If we allow $\C \subseteq \powerset W$ to be an arbitrary family of sets
then our semantics is equivalent to the neighborhood semantics that has
already been used in the literature
\cite{MartiPinosio14a,Negri15,Girlando21}. Thus the semantics in convex
geometries specializes the neighborhood semantics for preferential
conditional logic.
To see why our semantics specializes neighborhood semantics let us
dualize the semantic clause such that it is expressed in terms of the
family $\F$ of feasible sets. It then becomes the clause
\begin{tabbing}
 $M \models \varphi \cond \psi$ \quad iff \quad \= for all $F \in \F$
with $F \cap \ext{\varphi} \neq \emptyset$ there is a $G \in \F$ \\
 \> with $G \subseteq F$ and $G \cap \ext{\varphi} \neq \emptyset$ such
that $G \cap \ext{\varphi} \subseteq \ext{\psi}$.
\end{tabbing}
This clause is precisely the same as the clause that is used for
arbitrary families $\F \subseteq \powerset W$ in
\cite{MartiPinosio14a,Negri15,Girlando21}. It can be traced back to much
earlier approaches in premise semantics
\cite{Veltman76,Kratzer81,Benthem11} and can also be seen as the
generalization of the clause from \cite{Girard07} to the infinite case.

\begin{remark} \label{not one-step}
 By making the convex geometry in a model depending on the world of
evaluation, one can extend our semantics to deal with nested
conditionals. This means that we would consider models of the form $M =
(W,\C,V)$, where $\C : W \to \powerset \powerset W$ is such that $\C(w)$
is a convex geometry for all $w \in W$. The conditional is then
evaluated relative to a world $w$ by using the above clause relative to
the convex geometry $\C(w)$. In \cite{Negri15,Girlando21} this kind of
semantics is used, however, with the dualized semantic clause and without
requiring that $\C(w)$ is a convex geometry.
\end{remark}

\begin{remark} \label{impossible worlds}
 Observe that if in some model $M = (W,\C,V)$ we have that
$\ext{\varphi} \subseteq C$ for all $C \in \C$ then $M \models \varphi
\cond \bot$. In this sense the worlds in $\bigcap \C$ can be thought of
as impossible worlds. We do not require that $\emptyset \in \C$ because
we want to allow $W$ to contain such impossible worlds. For the results
of this paper this is not crucial because, as we argue in
Proposition~\ref{getting rid of impossible worlds} below, impossible
worlds can always be eliminated from $W$, without changing the set of
true conditionals. In more complex settings, such as the nested
semantics from Remark~\ref{not one-step} or the duality results from
\cite{MartiPinosio19}, it is however convenient to allow for impossible
worlds.
\end{remark}

If the antecedent of a conditional $\varphi \cond \psi$ evaluates to a
finite set $\ext{\varphi}$ then the semantic clause for the conditional
can be simplified.
\begin{proposition} \label{simpler clause}
 For any model $M = (W,\C,V)$ and $\varphi, \psi \in \Lang_0$ such that
$\ext{\varphi} \subseteq W$ is finite the following are equivalent
\begin{enumerate}
 \item $M \models \varphi \cond \psi$, \label{standard clause}
 \item $\extreme{\ext{\varphi}} \subseteq \ext{\psi}$, and
\label{extreme clause}
 \item $\ext{\varphi} \subseteq \hull{\ext{\varphi} \cap \ext{\psi}}$.
\label{closure clause}
\end{enumerate}
\end{proposition}
\begin{proof}
 The equivalence of items (\ref{extreme clause})~and~(\ref{closure
clause}) follows from Theorem~\ref{smooth theorem}. Hence, it suffices
to show that items (\ref{standard clause})~and~(\ref{closure clause})
are equivalent.

 Assume that $M \models \varphi \cond \psi$ and consider any $C \in
\C$ such that $\ext{\varphi} \cap \ext{\psi} \subseteq C$. We want to
show that then $\ext{\varphi} \subseteq C$. If this was not the case
then it would follow from $M \models \varphi \cond \psi$ that there
is some $D \in \C$ with $C \subseteq D$ such that $\ext{\varphi}
\nsubseteq D$ and $\ext{\varphi} \subseteq D \cup \ext{\psi}$. These
latter two inclusions entail that $\ext{\varphi} \cap \ext{\psi}
\nsubseteq D$, contradicting $\ext{\varphi} \cap \ext{\psi} \subseteq C
\subseteq D$.

 For the other direction assume that $\ext{\varphi} \subseteq
\hull{\ext{\varphi} \cap \ext{\psi}}$. We derive a contradiction from
the assumption that not $M \models \varphi \cond \psi$. The goal is to
construct an infinite, strictly increasing chain $C_0 \subset C_1
\subset \dots$ of convex sets such that $C_i \nsubseteq \ext{\varphi}$
and $(C_{i + 1} \setminus C_i) \cap \ext{\varphi} \neq \emptyset$ for
all $i \in \N$. This then contradicts the assumption that
$\ext{\varphi}$ is finite.

Because we assume that not $M \models \varphi \cond \psi$ there is some
$C \in \C$ with $C \nsubseteq \ext{\varphi}$ such that for every $D \in
\C$ with $C \subseteq D$ and $D \nsubseteq \ext{\varphi}$ we have that
$\ext{\varphi} \nsubseteq D \cup \ext{\psi}$. Let $C_0 = C$.

To construct $C_{i + 1}$ from $C_i$ assume that we have a $C_i \in \C$
such that $C_i \nsubseteq \ext{\varphi}$. From the assumption that
$\ext{\varphi} \subseteq \hull{\ext{\varphi} \cap \ext{\psi}}$ it
follows that there is some $x \in \ext{\varphi} \cap \ext{\psi}$ such
that $x \notin C_i$. Because $C = C_0 \subseteq C_i$ we obtain from the
choice of $C$ that $\ext{\varphi} \nsubseteq D \cup \ext{\psi}$. Thus,
there is some $y \in \ext{\varphi}$ such that $y \notin C_i$ and $y
\notin \ext{\psi}$. Because $x \in \ext{\psi}$ it follows that $x \neq
y$ and thus we can apply the anti-exchange property to obtain a convex
set $C^+$ with $C_i \subseteq C^+$ that contains precisely one of $x$
and $y$. We set $C_{i + 1} = C^+$. Since both $x$ and $y$ are in
$\ext{\varphi}$, but none of them is in $C_i$, it follows that $C_{i +
1} \nsubseteq \ext{\varphi}$ and $(C_{i + 1} \setminus C_i) \cap
\ext{\varphi} \neq \emptyset$.
\end{proof}

\begin{example}
 The picture in Figure~\ref{intro example} can be taken to show the
model $M = (W,\C,V)$ with
\begin{itemize}
 \item $W = \{x,y,z,u,v\} \subseteq \R^2$ with $x = (0,5)$, $y =
(4,5)$, $z = (2.4,3)$, $u = (1.9,4.3)$, $v = (1.2,1.5)$,
 \item $\C$
is the relative convexity of $W$ in $\R^2$, and
 \item $V(p) = \{x,z,u\}$, $V(q) = \{x,y,u\}$ and $V(r) = \{x,y,z\}$.
\end{itemize}
\end{example}

\begin{example} \label{running example}
 As the running example for our completeness proof we use the following
formula
\begin{equation*}
 \alpha = (\top \cond p) \land (q \cond p) \land (\neg (p
\leftrightarrow q) \cond p) \land \neg (\neg q \cond p) \land \neg ((p
\leftrightarrow q) \cond p) \land \neg (\neg p \cond \neg q).
\end{equation*}
A relatively simple model $M = (W,\C,V)$ in which $\alpha$ is true is as
follows:
\begin{itemize}
 \item $W = \{pq,p\bar{q},\bar{p}q,\bar{p}\bar{q}\}$ is a four
element set,
 \item $\C =
\{\emptyset,\{\bar{p}q\},\{\bar{p}\bar{q}\},\{pq,\bar{p}q\},\{p\bar{q},\bar{p}q\},\{\bar{p}q,\bar{p}\bar{q}\},W
\setminus \{p\bar{q}\},W \setminus \{pq\},W\}$, and
 \item $V(p) = \{pq,p\bar{q}\}$ and $V(q) = \{pq,\bar{p}q\}$.
\end{itemize}
\end{example}

\begin{example} \label{preferential models}
 Every model in the order semantics of the form $M = (W,\leq,V)$, where
$\leq$ is a partial order over $W$, yields a model $M' =
(W,\upsets{\leq},V)$ in the sense defined here. In fact $M$ and $M'$
satisfy the same conditionals. In the finite case this follows from the
reformulation of our semantic clause in Proposition~\ref{simpler clause}
and the observation that the minimal elements of some set in a poset are
precisely its extreme points in the upset convexity. In the infinite
case we leave it to the reader to check that the semantic clause for the
conditional relative to an infinite partial order $\leq$ from
\cite{Burgess81,Veltman85}
\begin{tabbing}
 $M \models \varphi \cond \psi$ \quad iff \quad \= for all $w \in
\ext{\varphi}$ there is a $v \leq w$ with $v \in \ext{\varphi}$ \\
\> such that for all $u \leq v$ if $u \in \ext{\varphi}$ then $u \in
\ext{\psi}$.
\end{tabbing}
is equivalent to the semantic clause given above with respect to the
upset convexity $\upsets{\leq}$. This connection between the order
semantics and the semantics in abstract convex geometries has as a
precursor the connection between the order semantics and premise
semantics that was already observed in
\cite{Lewis81,Benthem11,MartiPinosio14a}.
\end{example}

\section{Completeness for abstract convex geometries}
\label{completeness}

This section contains a completeness result for preferential conditional
logic with respect to the models from section~\ref{abstract semantics}
that are based on abstract convex geometries. It reads at follows:
\begin{maintheorem} \label{abstract completeness}
 Every one-step formula $\varphi \in \Lang_1$ that is consistent in
preferential conditional logic is true in a model of the form
$(W,\C,V)$, where $W$ is a finite set and $\C$ a convex geometry over
$W$.
\end{maintheorem}
This theorem is a consequence of at least two results that already exist
in the literature:
\begin{enumerate}
 \item Theorem~\ref{abstract completeness} can be obtained from the
well-know completeness with respect to the semantics in posets
\cite{Burgess81,Veltman85} together with the observation from
Example~\ref{preferential models} that every model based on a poset
gives rise to a model based on a convex geometry that satisfies the same
formulas. However, it needs to be checked that the necessary formal
proofs go through with our more restrictive one-step proof system and
that the completeness construction yields a finite model with an
anti-symmetric ordering.

 \item An alternative approach is to connect to the nonmonotonic
consequence relations from \cite{Kraus90} and then apply the duality
result from \cite{MartiPinosio19}. Observe that every consistent formula
$\varphi \in \Lang_1$ gives rise to a nonmonotonic consequence relation
$\quickDefaults$ satisfying the axioms of System~P, by taking $\alpha
\quickDefaults \beta$ iff $\vdash \varphi \rightarrow (\alpha \cond
\beta)$. If one then moves to the free Boolean algebra over $\Prop$,
which we can assume to be finite, then one is precisely on the algebraic
side of the dual correspondence from \cite{MartiPinosio19}. On the
spatial side of this duality one then obtains a convex geometry over the
atoms of the free Boolean algebra on $\Prop$.
\end{enumerate}
For readers who are not comfortable with adapting these existing results
we give a direct proof of Theorem~\ref{abstract completeness}.

\newcommand{\chara}[1]{{\chi({#1})}}

To prove Theorem~\ref{abstract completeness} we need to define a finite
model $M = (W,\C,V)$ such that $M \models \varphi$. We first discuss the
definition of the domain $W$ and the valuation $V : \Prop \to \powerset
W$. We let $W$ be the set of all assignments $a : \Prop \to \{0,1\}$ in
the sense of classical propositional logic. This set is finite because
we can assume $\Prop$ to be finite since there are only finitely many
propositional letters occurring in $\varphi$. The valuation $V : \Prop
\to \powerset W$ is defined such that $V(p) = \{a \in W \mid a(p) = 1\}$
for all $p \in \Prop$. By the completeness theorem for classical
propositional logic we have that $\ext{\alpha}_V \subseteq
\ext{\beta}_V$ iff $\vdash \alpha \rightarrow \beta$ for all
$\alpha,\beta \in \Lang_0$. We use this fact in the continuation of this
proof without explicitly mentioning it. We also need that for every set
$Y \subseteq W$ there is a characteristic formula $\chara{Y} \in
\Lang_0$ such that $\ext{\chara{Y}}_V = Y$. Because $\Prop$ and $W$ are
finite we can define $\chara{Y} = \bigvee_{a \in Y} \chara{a}$, where
$\chara{a} = \bigwedge \{p \mid a(p) = 1\} \land \bigwedge \{\neg p \mid
a(p) = 0\}$.

To define the convex geometry $\C$ we first fix a maximally consistent
set $\Sigma \subseteq \Lang_1$ with $\varphi \in \Sigma$. Because
$\varphi$ is consistent such a set exists by Lindenbaum's Lemma. Below
we are implicitly going to make use of the fact that $\Sigma$ is closed
under provable implications, that is, if $\vdash \bigwedge \Sigma'
\rightarrow \rho$ for some finite $\Sigma' \subseteq \Sigma$ then $\rho
\in \Sigma$. We then define the family of convex sets as follows:
\[
 \C = \{C \subseteq W \mid \ext{\alpha} \cap \ext{\beta} \subseteq C
\mbox{ implies } \ext{\alpha} \subseteq C \mbox{ for all } \alpha,\beta
\in \Lang_0 \mbox{ with } \alpha \cond \beta \in \Sigma \}.
\]
Define the model $M = (W,\C,V)$. To finish the proof of
Theorem~\ref{abstract completeness} we need to verify that $\C$ is a
convex geometry and that $M \models \varphi$. It is straight-forward to
check that $\C$ is closed under intersections. Thus it follows from
Lemma~\ref{has anti exchange} below, which states that $\C$ has the
anti-exchange property, that $\C$ is a convex geometry. That $M \models
\varphi$ follows from Lemma~\ref{truth lemma}, which states that $M
\models \theta$ iff $\theta \in \Sigma$ for all $\theta \in \Lang_1$.

\newcommand{\shullSym}{h}
\newcommand{\shull}[1]{{\shullSym({#1})}}

To prove Lemmas \ref{has anti exchange}~and~\ref{truth lemma} we need
the following syntactic characterization of the convex hull operator in
$\C$:
\begin{align*}
 \shullSym : \powerset W & \to \powerset W, \\
 Y & \mapsto \bigcup \{\ext{\delta} \subseteq W \mid \delta \cond
\chara{Y} \in \Sigma\}.
\end{align*}
It is possible to show that $\shullSym$ is the closure operator
associated to the meet semilattice $\C \subseteq \powerset W$. We do not
do this here because the completeness proof only needs the following
weaker properties of $\shullSym$:
\begin{lemma} \label{shull properties}
 For all $Y \subseteq W$ it holds that
\begin{enumerate}
 \item \label{extensive} $Y \subseteq \shull{Y}$, and
 \item \label{convex} $\shull{Y} \in \C$.
\end{enumerate}
\end{lemma}
\begin{proof}
 For item~\ref{extensive} observe that by (Id) we have that $\vdash
\chara{Y} \cond \chara{Y}$. Thus $\vdash \varphi \rightarrow (\chara{Y}
\cond \chara{Y})$ and $\chara{Y} \cond \chara{Y} \in \Sigma$, which
entails $Y = \ext{\chara{Y}} \subseteq \shull{Y}$ by the definition of
$\shullSym$.

 For item~\ref{convex} take any $\alpha,\beta \in \Lang_0$ such that
$\alpha \cond \beta \in \Sigma$ and $\ext{\alpha} \cap \ext{\beta}
\subseteq \shull{Y}$. We need to show that then $\ext{\alpha} \subseteq
\shull{Y}$. Because $W$ is finite it follows from $\ext{\alpha} \cap
\ext{\beta} \subseteq \shull{Y}$ that there are finitely many
$\delta_1,\dots,\delta_n \in \Lang_0$ with $\ext{\alpha} \cap
\ext{\beta} \subseteq \ext{\delta_1} \cup \dots \cup \ext{\delta_n}$ and
$\delta_i \cond \chara{Y} \in \Sigma$ for all $i \in \{1,\dots,n\}$.
From the former we get that $\vdash (\alpha \land \beta) \rightarrow
(\delta_1 \lor \dots \lor \delta_n)$. Using (RW) we obtain $\alpha \cond
(\delta_1 \lor \dots \lor \delta_n) \in \Sigma$ because by (Id) and
(And) we have that $\alpha \cond \alpha \land \beta \in \Sigma$. From
the latter, that $\delta_i \cond \chara{Y} \in \Sigma$ for all $i \in
\{1,\dots,n\}$, it follows with finitely many applications of (Or) that
$\delta_1 \lor \dots \lor \delta_n \cond \chara{Y} \in \Sigma$. Because
of the (CCut') from Proposition~\ref{derivable} we get that $\alpha \lor
\delta_1 \lor \dots \lor \delta_n \cond \chara{Y} \in \Sigma$. By the
definition of $\shullSym$ this entails $\ext{\alpha} \subseteq
\shull{Y}$.
\end{proof}

\begin{lemma} \label{has anti exchange}
 $\C$ has the anti-exchange property.
\end{lemma}
\begin{proof}
Consider any $C \in \C$ and $x \neq y$ with $x,y \notin C$. We derive a
contradiction from the assumption that for all $D \in \C$
with $C \subseteq D$ we have $x \in D$ iff $y \in D$.

If this assumption was true then it follows that $y \in \shull{\chara{C
\cup \{x\}}}$ because $x \in \shull{\chara{C \cup \{x\}}}$, $C \subseteq
\shull{\chara{C \cup \{x\}}}$ and $\shull{\chara{C \cup \{x\}}} \in \C$.
Thus there is some $\delta_y \in \Lang_0$ such that $y \in
\ext{\delta_x}$ and $\delta_y \cond \chara{C \cup \{x\}} \in \Sigma$.
Because $\vdash \chara{C \cup \{x\}} \rightarrow \chara{C \cup \{x,y\}}$
it follows from the derived rule (R) in Proposition~\ref{derivable} that
$(\delta_y \land \chara{C \cup \{x,y\}}) \lor \chara{C \cup \{x\}} \cond
\chara{C \cup \{x\}} \in \Sigma$. One can check that $(\ext{\delta_y}
\cap (C \cup \{x,y\})) \cup (C \cup \{x\}) = C \cup \{x,y\}$. Thus it
follows with (LLE) that $\chara{C \cup \{x,y\}} \cond \chara{C \cup
\{x\}} \in \Sigma$.

If we interchange the roles of $x$ and $y$ in the reasoning from the
previous paragraph we obtain that also $\chara{C \cup \{x,y\}} \cond
\chara{C \cup \{y\}} \in \Sigma$. Thus with the help of (And) we can
deduce $\chara{C \cup \{x,y\}} \cond (\chara{C \cup \{x\}} \land
\chara{C \cup \{y\}}) \in \Sigma$ from which we get $\chara{C \cup
\{x,y\}} \cond \chara{C} \in \Sigma$ by (RW). This contradicts $C \in
\C$ because $\ext{\chara{C \cup \{x,y\}}} \cap \ext{\chara{C}} \subseteq
C$ but $\ext{\chara{C \cup \{x,y\}}} \nsubseteq C$.
\end{proof}

\begin{lemma} \label{truth lemma}
 For all $\theta \in \Lang_1$ it holds that
\[
 M \models \theta \iff \theta \in \Sigma.
\]
\end{lemma}
\begin{proof}
The proof of this lemma is an induction on the complexity of $\Lang_1$.
The cases for the Boolean operators are straightforward. Thus we only
treat the base case where $\theta = \alpha \cond \beta$.

For the right-to-left direction assume that $\alpha \cond \beta \in
\Sigma$. To prove $M \models \alpha \cond \beta$ we show
that $\ext{\alpha} \subseteq \hull{\ext{\alpha} \cap \ext{\beta}}$,
where $\hullSym$ denotes the convex hull operator of the convex geometry
$\C$. Thus we need to show that $\ext{\alpha} \subseteq C$ for every
convex set $C \in \C$ with $\ext{\alpha} \cap \ext{\beta} \subseteq C$.
This follows directly from the definition of $\C$.

For the other direction assume that $M \models \alpha \cond \beta$. This
means that $\ext{\alpha} \subseteq \hull{\ext{\alpha} \cap
\ext{\beta}}$. Because by Lemma~\ref{shull properties}
$\shull{\ext{\alpha} \cap \ext{\beta}}$ is a convex set containing
$\ext{\alpha} \cap \ext{\beta}$ it follows that $\hull{\ext{\alpha} \cap
\ext{\beta}} \subseteq \shull{\ext{\alpha} \cap \ext{\beta}}$. Thus
$\ext{\alpha} \subseteq \shull{\ext{\alpha} \cap \ext{\beta}}$. Because
$W$ is finite it follows from the definition of $\shullSym$ that there
are $\delta_1,\dots,\delta_n$ such that $\ext{\alpha} \subseteq
\ext{\delta_1} \cup \dots \cup \ext{\delta_n}$ and $\delta_i \cond
\alpha \land \beta \in \Sigma$ for all $i \in \{1,\dots,n\}$. It follows
from the former with the help of (Id) and (RW) that $\vdash \alpha \cond
\delta_1 \lor \dots \lor \delta_n$. Using (Or) we get that $\delta_1 \lor \dots \lor \delta_n \cond \alpha \land \beta
\in \Sigma$ because $\delta_i \cond
\alpha \land \beta \in \Sigma$ for all $i \in \{1,\dots,n\}$. With the help of (CCut'), which is derivable according to
Proposition~\ref{derivable}, it follows that $\alpha \lor \delta_1 \lor
\dots \lor \delta_n \cond \alpha \land \beta \in \Sigma$. Because of
(WCM) from Proposition~\ref{derivable} we obtain that $(\alpha \lor
\delta_1 \lor \dots \lor \delta_n) \land \alpha \cond \beta \in \Sigma$
and by (LLE) we get that $\alpha \cond \beta \in \Sigma$.
\end{proof}

\begin{remark}
 Note that no two distinct worlds in the model that is constructed in
the proof of Theorem~\ref{abstract completeness} satisfy the same
propositional letters. This is in stark contrast to the completeness
proofs of preferential conditional logic with respect to its semantics
in orders from \cite{Burgess81} and \cite{Veltman85}. Part of the
complexity of the constructions in these proofs comes from the fact that
they duplicate possible worlds to obtain enough witnesses in the
constructed order. It follows from the discussion of the coherence
condition in Section~II.4.1 of \cite{Veltman85} or from the example in
the last paragraph of Section~5.2 in \cite{Kraus90} that such a
duplication of worlds is necessary to obtain completeness with respect
to the order semantics. That such a duplication of worlds is not needed
for completeness with respect to convex geometries is exploited in the
duality result from \cite{MartiPinosio19}, which uses convex geometries
on the spatial side of the duality.
\end{remark}

\section{Morphisms of convex geometries}
\label{morphisms}

In this section we recall the notion of a morphism between convex
geometries from \cite{MartiPinosio19}. The motivation for this notion is
that in the finite case they are precisely the functions that preserve
and reflect the truth of all conditionals. It should be mentioned that
our notion of morphism can not be straight-forwardly adapted to the
infinite case as its adequacy relies on the reformulation of the
semantics from Proposition~\ref{simpler clause}, which only holds in the
finite case.

The definition of a morphism uses the following existential and
universal image maps: For every $f : W \to U$ we write $\exiimage{f} :
\powerset W \to \powerset U$ for the left adjoint and $\uniimage{f} :
\powerset W \to \powerset U$ for the right adjoint of the inverse image
map $f^{-1} : \powerset U \to \powerset W, X \mapsto \{w \in W \mid f(w)
\in X\}$. Concretely, this means that for all $Y \subseteq W$
\begin{align*}
 \exiimage{f}(Y) & = \{u \in U \mid f^{-1}(\{u\}) \cap Y \neq
\emptyset\}, \mbox{ and} \\
 \uniimage{f}(Y) & = \{u \in U \mid f^{-1}(\{u\}) \subseteq Y\}.
\end{align*}
It is easy to check that $\comp{\exiimage{f}(Y)} = \uniimage{f}(\comp{Y})$
for all $Y \subseteq W$. Note that $\exiimage{f}$ is just the usual
direct image map.

A \emph{morphism} $f$ from a convex geometry $(W,\C)$ to a convex
geometry $(U,\D)$ is a function $f : W \to U$ such that $\uniimage{f}(C)
\in \D$ for all $C \in \C$. The morphism $f$ is a \emph{strong morphism}
if it additionally satisfies that for every $D \in \D$ there is some $C
\in \C$ such that $D = \uniimage{f}(C)$. Thus, strong morphism are
precisely the functions for which $\D = \{\uniimage{f}(C) \subseteq U
\mid C \in \C\}$. By dualizing and exploiting $\comp{\exiimage{f}(Y)} =
\uniimage{f}(\comp{Y})$ one can adapt this definition of morphism to the
feasible sets of a convex geometry. A morphism is then a function $f$
such that $\exiimage{f}(F)$ is feasible for every feasible $F$, and it
is strong if every feasible set arises as $\exiimage{f}(F)$ for some
feasible $F$.

The reader can convince themselves that surjective affine transformation
on the plane, such as translations, rotations or scalings, are strong
morphisms.

For posets we have that $f : W \to U$ is a morphism between the upset
convexities of partial orders $\leq$ on $W$ and $\leq'$ on $U$ if and
only if it satisfies the following condition, which is just the back
condition on bounded morphism in modal logic:
\begin{itemize}
 \item For all $w \in W$ and $u' \leq' f(w)$ there is a $u \leq w$ such
that $f(u) = u'$.
\end{itemize}
The morphism $f$ is strong if and only if it additionally satisfies the
following condition:\footnote{In \cite{MartiPinosio19} we made the false
claim that the strong morphism between posets are the order preserving
and surjective functions.}
\begin{itemize}
 \item For all $u \in U$ there is a $w \in W$ such that $f(w) = u$ and
for all $w' \leq w$ we have $f(w') \leq' u$.
\end{itemize}
Note that these two conditions on the graph of $f$ correspond to the
conditions on bisimulations between models based on posets from
\cite{Zhu06}.

A further example of a morphism comes from the following proposition. It
shows that removing impossible worlds from a model yields a submodel
that embeds with a strong morphism. As a consequence impossible worlds
can be removed without altering the truth of one-step formulas.
\begin{proposition} \label{getting rid of impossible worlds}
 Let $(W,\C)$ be any convex geometry and let $I = \bigcap \C$. Define $U
= W \setminus I$ and let $\D$ be the relative convexity of $U$ in $W$.
Then $\emptyset \in \D$ and the embedding $e : U \to W,u \mapsto u$ is a
strong morphism from $(U,\D)$ to $(W,\C)$.
\end{proposition}
\begin{proof}
 That $\emptyset \in \D$ follows because, by the closure of $\C$ under
arbitrary intersection we have that $I \in \C$, and thus $\emptyset = I
\cap U \in \D$ by the definition of the relative convexity. To see that
$e$ is a strong morphism it is easier to reason with the feasible sets.
The worlds in $I$ do not appear in any feasible set from $(W,\C)$ and
thus it is clear that the feasible sets in $(W,\C)$ are precisely the
direct images of feasible sets from $(U,\D)$.
\end{proof}

We can lift the notion of a morphism to models in the standard way.
That is, $f : W \to W'$ is a morphism from $M = (W,\C,V)$ to $M' =
(W',\C',V')$ if $f$ is a morphism from $(W,\C)$ to $(W',\C')$ and $V(p)
= f^{-1}(V'(p))$ for all $p \in \Prop$. We call $f$ from $M$ to $M'$
strong if it is strong as a morphism between the underlying convex
geometries.

Propositions 10~and~12 from \cite{MartiPinosio19} entail that in the
finite case strong morphisms preserve and reflect the truth of
conditionals. Because this result is central for our approach we restate
the result in our terminology and provide a self contained proof.
\begin{theorem} \label{adequacy}
 Let $f$ be a strong morphism from a finite model $M = (W,\C,V)$ to a
finite model $M' = (W',C',V')$ then it holds for all $\varphi \in
\Lang_1$ that
\[
 M \models \varphi \iff M' \models \varphi.
\]
\end{theorem}
\begin{proof}
 First observe that because taking preimages is a Boolean homomorphism
between powerset algebras it is clear that the condition that $V(p) =
f^{-1}(V'(p))$ for all $p \in \Prop$ entails that $\ext{\varphi}_V =
f^{-1}(\ext{\varphi}_{V'})$ for all $\varphi \in \Lang_0$.

To prove the preservation of true formulas in $\Lang_1$ one uses a
standard induction on the complexity of formulas. We only consider the
case for the conditional. Because we are in a finite setting we can use
the equivalent formulation of the semantics from
Proposition~\ref{simpler clause}, stating that $\varphi \cond \psi$ is
true in a model iff $\ext{\varphi} \subseteq \hull{\ext{\varphi} \cap
\ext{\psi}}$.

Assume first that $\ext{\varphi}_{V'} \subseteq \hull{\ext{\varphi}_{V'}
\cap \ext{\psi}_{V'}}$ holds in $\C'$. To show that then
$f^{-1}(\ext{\varphi}_{V'}) \subseteq \hull{f^{-1}(\ext{\varphi}_{V'})
\cap f^{-1}(\ext{\psi}_{V'})}$ holds in $\C$ consider any $w \notin
\hull{f^{-1}(\ext{\varphi}_{V'}) \cap f^{-1}(\ext{\psi}_{V'})}$. This
means that there is some convex $C \in \C$ such that
$f^{-1}(\ext{\varphi}_{V'} \cap \ext{\psi}_{V'}) =
f^{-1}(\ext{\varphi}_{V'}) \cap f^{-1}(\ext{\psi}_{V'}) \subseteq C$ and
$w \notin C$. Because $f$ is a morphism it then follows that
$\uniimage{f}(C) \in \C'$ and because $\uniimage{f}$ is right adjoint to
$f^{-1}$ we get $\ext{\varphi}_{V'} \cap \ext{\psi}_{V'} \subseteq
\uniimage{f}(C)$. Thus, $\hull{\ext{\varphi}_{V'} \cap \ext{\psi}_{V'}}
\subseteq \uniimage{f}(C)$ and with the assumption that
$\ext{\varphi}_{V'} \subseteq \hull{\ext{\varphi}_{V'} \cap
\ext{\psi}_{V'}}$ it follows that $\ext{\varphi}_{V'} \subseteq
\uniimage{f}(C)$. From $w \notin C$ we have that $f(w) \notin
\uniimage{f}(C)$ and so $f(w) \notin \ext{\varphi}_{V'}$, which means
that $w \notin f^{-1}(\ext{\varphi}_{V'})$.

For the other direction assume $f^{-1}(\ext{\varphi}_{V'}) \subseteq
\hull{f^{-1}(\ext{\varphi}_{V'}) \cap f^{-1}(\ext{\psi}_{V'})}$. We show
$\ext{\varphi}_{V'} \subseteq \hull{\ext{\varphi}_{V'} \cap
\ext{\psi}_{V'}}$ by contraposition. Thus consider any $w' \notin
\hull{\ext{\varphi}_{V'} \cap \ext{\psi}_{V'}}$. There then is some
convex $C' \in \C'$ such that $\ext{\varphi}_{V'} \cap \ext{\psi}_{V'}
\subseteq C'$ and $w' \notin C'$. Because $f$ is a strong morphism there
exists a $C \in \C$ such that $C' = \uniimage{f}(C)$. Because
$\uniimage{f}$ is right adjoint to $f^{-1}$ we obtain
$f^{-1}(\ext{\varphi}_{V'} \cap \ext{\psi}_{V'}) \subseteq C$ from
$\ext{\varphi}_{V'} \cap \ext{\psi}_{V'} \subseteq C' =
\uniimage{f}(C)$. With $f^{-1}(\ext{\varphi}_{V'} \cap \ext{\psi}_{V'})
= f^{-1}(\ext{\varphi}_{V'}) \cap f^{-1}(\ext{\psi})_{V'}$ it follows
that $\hull{f^{-1}(\ext{\varphi}_{V'}) \cap f^{-1}(\ext{\psi}_{V'})}
\subseteq C$. Combining with the assumption $f^{-1}(\ext{\varphi}_{V'})
\subseteq \hull{f^{-1}(\ext{\varphi}_{V'}) \cap
f^{-1}(\ext{\psi}_{V'})}$ yields $f^{-1}(\ext{\varphi}_{V'}) \subseteq
C$. Since $w' \notin \uniimage{f}(C)$ it must be the case that
$f^{-1}(\{w'\}) \nsubseteq C$, and hence $w \notin \ext{\varphi}_{V'}$.
\end{proof}

\section{Representation of convex geometries in the plane}
\label{representation}

In this section we show that the representation from Theorem~5 in
\cite{Richter17} gives rise to a strong morphism of convex geometries.
It would be possible to show that any such representation of a finite
convex geometry with polygons that have disjoint extreme points yields a
strong morphism. Thus, we could just use Theorem~5 from \cite{Richter17}
as a black box, without disassembling the inner workings of the
construction in its proof. But because this construction is at the heart
of our completeness result, we give a detailed exposition of the
representation in this section. Figure~\ref{pretty representation}
contains an example of this representation for the convex geometry from
Example~\ref{running example}.

\subsection{Decomposition of finite convex geometries}

It is shown in \cite{Edelman85} that every finite convex geometry can
be decomposed into a family of convexities arising from linear orders.
Using these decompositions is crucial for the results in
\cite{Richter17}.

The relevant notion of decomposition is the join in the semi-lattice of
all convex geometries over some fixed finite set $W$, ordered by the
inclusion between sets of sets. From Theorem~2.2 in \cite{Edelman80} it
follows that the join $\C \vee \D$ of convex geometries $\C$ and $\D$
over $W$ can be defined concretely as
\[
 \C \vee \D = \{C \cap D \mid C \in \C \mbox{ and } D \in \D\}.
\]

Recall that a partial order $\leq$ on $W$ is \emph{linear} if $x \leq y$
or $y \leq x$ holds for all $x,y \in W$. The decomposition result, which
is Theorem~5.2 in \cite{Edelman85}, can be formulated in our notation as
follows:
\begin{theorem} \label{decomposition theorem}
 Let $\C$ be a convex geometry over a finite set $W$ such that
$\emptyset \in \C$. Then there is a finite family of linear orders
$(\leq_j)_{j = 1}^m$ such that
\begin{equation} \label{decomposition}
 \C = \bigvee_{j = 1}^m \upsets{\leq_j}.
\end{equation}
\end{theorem}

Note that from the definition of the join it follows that if the posets
$(\leq_i)_{i = 1}^k$ are a decomposition of a convex geometry $\C$
according to \eqref{decomposition} then some set $X \subseteq W$ is
convex if and only if it can be written as
\begin{equation} \label{convex as upsets}
 X = \bigcap_{j = 1}^{m} \upseti{j}{X},
\end{equation}
where $\upseti{j}{X} = \{w \in W \mid x \leq_j w \mbox{ for some }x \in
X\}$ denotes the upwards closure of $X$ in the order $\leq_j$.

\begin{figure}
\begin{center}
\begin{tikzpicture}[scale=1.4]
  \node (innerup) at (3.8,3) {
\begin{tikzpicture}[scale=0.7]
 \node (w1) at (0,0) {$p \bar{q}$};
 \node (w2) at (0,1.5)  {$\bar{p} \bar{q}$};
 \node (w3) at (0,3)  {$p q$};
 \node (w4) at (0,4.5) {$\bar{p} q$};

 \node (u1) at (1.8,0) {$p q$};
 \node (u2) at (1.8,1.5)  {$\bar{p} \bar{q}$};
 \node (u3) at (1.8,3)  {$p \bar{q}$};
 \node (u4) at (1.8,4.5) {$\bar{p} q$};

 \node (v1) at (3.6,0) {$p q$};
 \node (v2) at (3.6,1.5)  {$p \bar{q}$};
 \node (v3) at (3.6,3)  {$\bar{p} q$};
 \node (v4) at (3.6,4.5) {$\bar{p} \bar{q}$};

 \draw (w1) -- (w2) -- (w3) -- (w4);
 \draw (u1) -- (u2) -- (u3) -- (u4);
 \draw (v1) -- (v2) -- (v3) -- (v4);
\end{tikzpicture}
};
  \node (innerdown) at (3.8,-3) {
$\begin{array}{l}
 \alpha = (\top \cond p) \land 
  (q \cond p) \land \phantom{d} \\
  (\neg (p \leftrightarrow q) \cond p) \land 
  \neg (\neg q \cond p) \land \phantom{d} \\
  \neg ((p \leftrightarrow q) \cond p) \land 
  \neg (\neg p \cond \neg q)
 \end{array}$
};
  \draw [dotted] (-2.7,0) -- (5,0);
  \draw [dotted] (0,-4.3) -- (0,4.3);

  \draw[->] (0,0) -- (0:4.5);
  \draw[->] (0,0) -- (120:4.5);
  \draw[->] (0,0) -- (240:4.5);


 \node[fill=white] (w1) at (0:4) {$p \bar{q}$};
 \node[fill=white] (w2) at (0:3) {$\bar{p} \bar{q}$};
 \node[fill=white] (w3) at (0:2) {$p q$};
 \node[fill=white] (w4) at (0:1) {$\bar{p} q$};

 \node[fill=white] (u1) at (120:4) {$p q$};
 \node[fill=white] (u2) at (120:3) {$\bar{p} \bar{q}$};
 \node[fill=white] (u3) at (120:2) {$p \bar{q}$};
 \node[fill=white] (u4) at (120:1) {$\bar{p} q$};

 \node[fill=white] (v1) at (240:4) {$p q$};
 \node[fill=white] (v2) at (240:3) {$p \bar{q}$};
 \node[fill=white] (v3) at (240:2) {$\bar{p} q$};
 \node[fill=white] (v4) at (240:1) {$\bar{p} \bar{q}$};
\end{tikzpicture}
\end{center}
\caption{In the upper right corner is a decomposition of the convex
geometry from Example~\ref{running example} into linear orders. The main
picture contains the representation of this convex geometry in the
plane. In the lower right corner is the formula $\alpha$ from
Example~\ref{running example} that is true in this convex geometry.}
\label{pretty representation}
\end{figure}
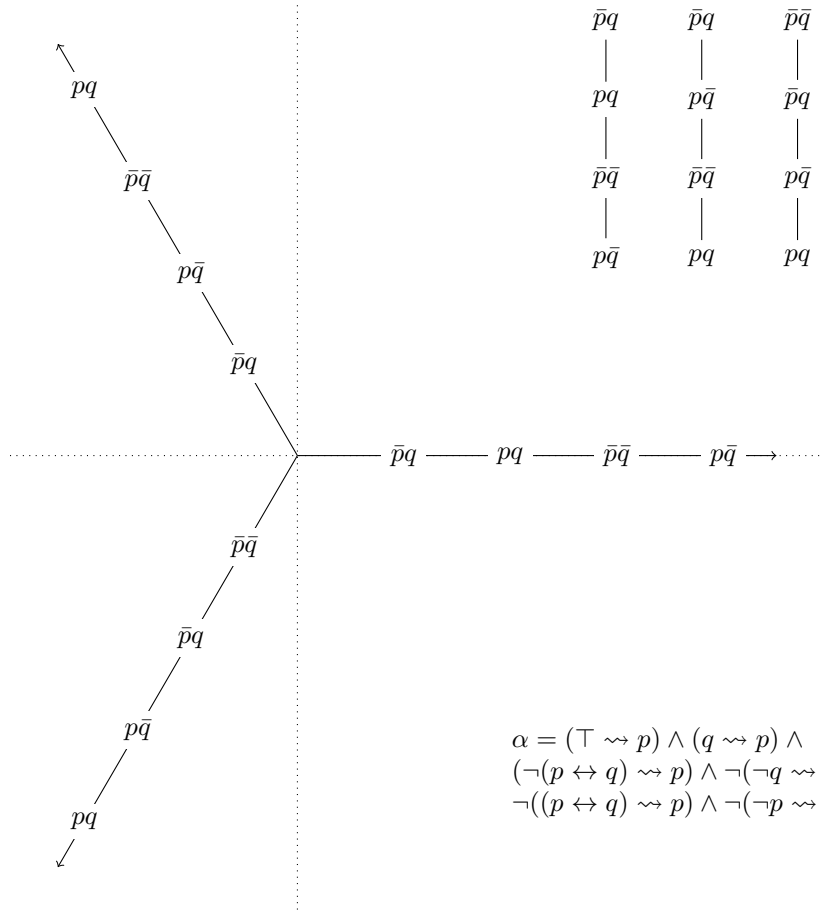


\subsection{The representation by Richter and Rogers}

\newcommand{\point}[2]{{u({#2},{#1})}}
\newcommand{\rank}[2]{{r_{#2}({#1})}}

This subsection contains the proof of Theorem~5 in \cite{Richter17}. For
this paper we need following formulation of the representation result:
\begin{maintheorem} \label{representation theorem}
 Let $\C$ be a convex geometry over a finite set $W$ such that
$\emptyset \in \C$. Then there is a finite set $U \subseteq \R^2$ and a
strong morphism of convex geometries $r$ from $(W,\C)$ to $U$
with the relative convexity from $\R^2$.
\end{maintheorem}

We first describe how to construct the set $U$ and the function $r$. Fix
a convex geometry $(W,\C)$ such that $\emptyset \in \C$ and let $n$ be
the number of elements in $W$. By Theorem~\ref{decomposition theorem}
there exists a decomposition of $\C$ into linear orders $(\leq_i)_{i =
1}^m$. Assume without loss of generality that $m \geq 2$, otherwise just
duplicate one of the linear orders. For every $w \in W$ and $j \in
\{1,\dots,m\}$ let $\rank{w}{j} \in \N \subseteq \R$ be the rank of $w$
in the linear order $\leq_j$ starting from the top. This means that if
$\leq_j$ is $w_n <_j w_{n - 1} <_j \dots <_j w_1$ then $\rank{w}{j} = i$
for the unique $i$ with $w_i = w$.

We then choose $m$-many points on the unit circle that are equally
distributed among all directions. Thus, set $d_j = (\cos(2\pi
j/m),\sin(2\pi j/m)) \in \R^2$ for every $j \in \{1,\dots,m\}$. Define
$s \in \R$ as
\[
 s = \max\left\{0,\frac{n \cos(2\pi/m)}{1 - \cos(2\pi/m)} \right\}.
\]

For every $w \in W$ and $j \in \{1,\dots,m\}$ define the point \[
\point{w}{j} = (s + \rank{w}{j}) d_j \in \R^2, \] and for every $w \in
W$ define the set $U^w = \{\point{w}{1} ,\dots,\point{w}{m}\}$. Clearly
we have that $U^w \cap U^u = \emptyset$ whenever $w \neq u$. Define $U
\subseteq \R^2$ as $U = \bigcup_{w \in W} U^a$ and $r : U \to W$ such
that $r(u)$ is the unique $w \in W$ with $u \in U^w$. Note that
$r^{-1}(\{w\}) = U^w$ for all $w \in W$.

The idea behind the definition of $U$ is to spread out the linear orders
in the decomposition of $\C$ along separate rays that move outwards from
the origin. On each ray this happens at distance $s$ away from the
origin. This safety distance ensures that every point on some ray is
further out from the origin than the intersection of the ray with any
line segment between points on neighboring rays.

Theorem~\ref{representation theorem} then follows from the following two lemmas.
\begin{lemma}
 $r$ is a morphism of convex geometries.
\end{lemma}
\begin{proof}
We need to show that whenever $C \subseteq U$ is convex in the relative
convexity of $U$ in $\R^2$ then $\uniimage{r}(C) \in \C$. Thus fix such
a $C$ and let $D = \uniimage{r}(C)$. To show that $D$ is convex in $\C$
we use the characterization \eqref{convex as upsets} and show that $D =
\bigcap_{j = 1}^{m} \upseti{j}{D}$. For the non-trivial
$\supseteq$-inclusion consider any $w$ such that for all $j \in
\{1,\dots,m\}$ there is some $w_j \in D$ such that $w_j \leq_j w$. To
prove $w \in D = \uniimage{r}(C)$ we need to show that $U^w \subseteq
\hull{C}$.

First observe that the origin $(0,0)$ is in the convex hull $\hull{C}$
of $C$ in $\R^2$. This is a little technical but not very interesting:
If $n$ is even then the origin can be written as a convex combination of
the points $\point{m}{w_m}$ and $\point{m/2}{w_{m/2}}$ in $C$ because
both points have $0$ in their second coordinate, and the former has a
positive but the latter a negative first coordinate. If $n$ is odd then
$n \geq 3$ and the points $\point{j}{w_j}$ and $\point{k}{w_k}$, for $j
= (m - 1)/2$ and $k = (m + 1)/2$, are in $C$. They both have a negative
first coordinate and a different signum in their second coordinates.
Thus there is some point $s \in \hull{C}$ that has a negative first
coordinate and $0$ in the second coordinate. The origin is then a convex
combination of $s$ and $\point{m}{w_m}$.

Consider then any point in $U^w$, which must be of the form
$\point{j}{w}$ for some $j \in \{1,\dots,m\}$. Because $D =
\uniimage{r}(C)$ and $w_j \in D$ we have that $\point{j}{w_j} \in
U^{w_j} = f^{-1}(\{w_j\}) \subseteq C$. Moreover, from $w_j \leq w$ it
follows that $\rank{w}{j} \leq \rank{w_j}{j}$ and hence $\point{j}{w} =
(s + \rank{w}{j}) d_j$ is on the line segment from the origin to
$\point{j}{w_j} = (s + \rank{w_j}{j}) d_j$. It follows that
$\point{j}{w} \in \hull{C}$ and thus that $\point{j}{w} \in C$, since
$C$ is convex in the relative convexity.
\end{proof}

\begin{lemma}
 $r$ is a strong morphism of convex geometries.
\end{lemma}
\begin{proof}
To show that $r$ is strong consider any $D \in \C$. We show that $D =
\uniimage{r}(C)$ for $C = \hull{r^{-1}(D)} \cap U$. That $D \subseteq
\uniimage{r}(C)$ follows immediately from $r^{-1}(D) \subseteq C$. To
show $D \supseteq \uniimage{r}(C)$ consider any $w \notin D$. We show
that $w \notin \uniimage{r}(C)$.

Because $D$ is convex we can apply the characterization from
\eqref{convex as upsets} and conclude that there is some $j \in
\{1,\dots,m\}$ such that $w <_j u$ for all $u \in D$. We then assume
that $j = m$. This is without loss of generality because one can apply a
rotation to turn any ray for $j$ until it comes to lie on the positive
$x$-axis. Because rotations are isomorphism with respect to the convex
sets this does not influence our reasoning.

To show that $w \notin \uniimage{r}(C)$ it suffices to show that
$\point{m}{w} \notin \hull{r^{-1}(D)}$. To this aim we show that the
first coordinate of $\point{m}{w} = (s + \rank{w}{m}) d_k$ is strictly
larger than the first coordinate of any $\point{k}{v} = (s +
\rank{v}{k}) d_k$ for $v \in D$ and $k \in \{1,\dots,m\}$, meaning that
$\point{m}{w}$ can not be written as the convex combination of such
points. If $k = m$ then this is clear because $d_m = (1,0)$ and
$\rank{w}{m} > \rank{v}{m}$, as $w <_m v$. In the other case where $k
\neq m$ first consider the case where $\cos(2\pi/m) \geq 0$. Then $0
\leq \frac{n \cos(2\pi/m)}{1 - \cos(2 \pi/m)} = s$ and we can estimate
the first coordinate of $\point{k}{v}$ as follows:
\begin{align*}
 (s + \rank{v}{k}) \cos(2\pi k / m) & \leq (s + n) \cos(2\pi k / m) \\
 & \leq (s + n) \cos(2\pi / m) \\
 & \leq \left(\frac{n\cos(2\pi / m)}{1 - \cos(2\pi /m )} + n\right) \cos(2\pi / m) \\
 & = \left(\frac{n}{1 - \cos(2\pi /m )}\right) \cos(2\pi / m) \\
 & \leq s \\
 & < s + \rank{w}{m}
\end{align*}
Because $s + \rank{w}{m}$ is the first coordinate of $\point{m}{w}$ this
is the needed inequality. In the other case where $\cos(2\pi / m) < 0$
we get that $m \leq 3$. Thus, $k / m$ is either $1/3$, $2/3$ or $1/2$
and so $\cos(2\pi k / m)$ is negative. It follows that the first
coordinate of $\point{k}{v}$ is also negative and therefore it is
smaller than the first coordinate of $\point{m}{w}$.
\end{proof}

\section{Completeness for Euclidean convexity}
\label{euclid}

In this last section we put the results from this paper together to
prove the completeness of preferential conditional logic with respect to
convexity between points in the plane. We also show that this result can
not be improved to a completeness result with respect to convexity on
the real line.

The following is the main result of this paper:
\begin{maintheorem} \label{completeness theorem}
 Every one-step formula $\varphi \in \Lang_1$ that is consistent in
preferential conditional logic is true in a model of the form $M =
(W,\C,V)$, where $W \subseteq \R^2$ is a finite set of points and $\C$
is the relative convexity of $W$ in $\R^2$.
\end{maintheorem}
\begin{proof}
 From Theorem~\ref{abstract completeness} we obtain a finite model $M''
= (W'',\C'',V'')$ such that $M'' \models \varphi$. From
Proposition~\ref{getting rid of impossible worlds} we get a finite convex
geometry $(W',\C')$ with $\emptyset \in \C'$ and strong morphism of
convex geometries $r''$ from $(W',\C')$ to $(W'',\C'')$. We can then
apply Theorem~\ref{representation theorem} to obtain a finite set $W
\subseteq \R^2$ together with a strong morphism $r'$ from $(W,\C)$ to
$(W',\C')$ such that $\C$ is the relative convexity of $W$ in $\R^2$.

Let $r = r' \circ r''$ be the composition of $r''$ with $r'$. Clearly,
this is also a strong morphism of convex geometries from $(W,\C)$ to
$(W'',\C'')$. Then define the model $M = (W,\C,V)$ such that $V(p) =
r^{-1}(V''(p))$ for all $p \in \Prop$. This turns $r$ into a strong
morphism from the model $M$ to the model $M''$ and thus $M \models
\varphi$ follows with Theorem~\ref{adequacy}.
\end{proof}

\begin{remark} \label{nested completeness}
 To adapt this completeness result to nested preferential conditional
logic one would need to consider models $(W,U,V)$ where $W \subseteq
\R^2$ and $U : W \to \powerset W$. The function $U$ fixes a finite set
of points $U(w)$ for every world $w \in W$. At a worlds $w \in W$ a
conditional is then evaluated in the relative convexity of $U(w)$ in
$\R^2$. Completeness with respect to such models can be obtained by
starting from a model in the semantics from Remark~\ref{not one-step}
and then applying Theorem~\ref{completeness theorem} locally to $\C(w)$
for every world $w$. By suitably translating the points in the sets
$U(w)$ one can ensure that $U(w) \cap U(w') = \emptyset$ whenever $w
\neq w'$. Thus, the valuation $V : \Prop \to \powerset W$ can be defined
globally on $W$.
\end{remark}

\label{line not enough}

The completeness result from Theorem~\ref{completeness theorem} can not
be improved to a completeness with respect to models based on subsets of
the real line. The reason is that such models validate additional
formulas that are not provable in preferential conditional logic. As a
first example consider the formula
\[
 \gamma_2 = (p \lor q \lor r \cond p \lor q) \lor (p \lor q \lor r \cond
p \lor r) \lor (p \lor q \lor r \cond q \lor r).
\]
It can be seen as a generalization of the formula $\gamma_1 = (p \lor q
\cond p) \lor (p \lor q \cond q)$, which is valid over linear orders.
Using soundness of the semantics over posets it is easy to see that
$\gamma_2$ is not derivable in preferential conditional logic. However,
one can show that $\gamma_2$ is true in all models of the form
$(W,\C,V)$, where $W \subseteq \R$ is finite and $\C$ is the relative
convexity of $W$ in $\R$. The argument is roughly that we just need to
consider the two propositional letters among $p$, $q$ and $r$ that are
true at the at most two extreme points of $\ext{p \lor q \lor r}$. Note
that these extreme points are simply the minimal and maximal elements of
$\ext{p \lor q \lor r}$ in the standard order of the reals.

Surprisingly, $\gamma_2$ can be invalidated if we allow $W$ to be an
infinite subset of $\R$. This shows that the conditional logic of finite
sets of points on the real line is different from the logic of the whole
real line. To invalidate $\gamma_2$ it suffices to consider a model
$(\R,\C,V)$, where $\C$ is the standard convexity on $\R$ and $V$ is
such that for every propositional letter in $\{p,q,r\}$ there are
arbitrarily large and arbitrarily small reals at which the propositional
letter is true.

The logic of infinite subsets of the real line is still stronger than
preferential conditional logic. To see this consider the formula
\[
 \delta_2 = (p \lor q \lor r \cond s) \rightarrow (p \lor q \cond s)
\lor (p \lor r \cond s) \lor (q \lor r \cond s).
\]
This formula is a generalization of the formula $\delta_1 = (p \lor q
\cond s) \rightarrow (p \cond s) \lor (q \cond s)$ expressing
disjunctive rationality, which is valid over interval orders. Using the
order semantics it is not hard to show that $\delta_2$ is not derivable
in preferential conditional logic. But $\delta_2$ is valid in models
that are based on the real line:
\begin{proposition}
 The formula $\delta_2$ is valid in all models of the form $(W,\C,V)$,
where $W \subseteq \R$ is any set of points on the line and $\C$ is the
relative convexity of $W$ in $\R$.
\end{proposition}
\begin{proof}
Consider a model $M = (W,\C,V)$ such that $W \subseteq \R$ and $\C$ is
the relative convexity of $W$ in $\R$. To show that $\delta_2$ is valid
assume that $M \models p \lor q \lor r \cond s$.

Define $p_r \in \{p,q,r\}$ such that for all $u \in \ext{p \lor q \lor
r}$ there is some $v \in \ext{p_r}$ with $u \leq v$. Such a $p_r$ must
exist. Otherwise, we have for all $a \in \{p,q,r\}$ a $u_a \in \ext{p
\lor q \lor r}$ such that $v < u_a$ for all $v \in \ext{a}$. This leads
to a contradiction by considering the maximum of $u_p$, $u_q$ and $u_r$,
which is in $\ext{p \lor q \lor r}$, but can not be in any of $\ext{p}$,
$\ext{q}$ and $\ext{r}$. Analogously, we define $p_l \in \{p,q,r\}$ such
that for all $u \in \ext{p \lor q \lor r}$ there is some $v \in
\ext{p_l}$ with $v \leq u$. Let $A = \{a_1,a_2\}$ be one of $\{p,q\}$,
$\{p,r\}$, or $\{q,r\}$ such that $\{p_r,p_l\} \subseteq A$.

We claim that then $M \models a_1 \lor a_2 \cond s$. To see this
consider any convex set $C \in \C$ such that $\ext{a_1 \lor a_2}
\nsubseteq C$. Thus, there is some world $u \in \ext{a_1 \lor a_2}$ such
that $u \notin C$. Because $C$ is convex it follows that the worlds in
$C$ are either all to the left or are all to the right of $u$. Assume
without loss of generality that all worlds of $C$ are to the left of
$u$, that is, $v < u$ for all $v \in C$. Let $C' = (-\infty,u)$ be the
convex set of all worlds that are strictly to the left of $u$. Clearly
$C \subseteq C'$ and $u \notin C'$. From the latter it follows that
$\ext{p \lor q \lor r} \nsubseteq C'$, because $u \in \ext{a_1 \lor a_2}
\subseteq \ext{p \lor q \lor r}$.

From the assumption that $M \models p \lor q \lor r \cond s$ it
follows that there is some convex set $D$ with $C' \cap \ext{p \lor q
\lor r} \subseteq D$ and $\ext{p \lor q \lor r} \nsubseteq D$ such that
$\ext{p \lor q \lor r} \subseteq D \cup \ext{s}$. From $C' \cap \ext{p
\lor q \lor r} \subseteq D$ it follows that $C \cap \ext{a_1 \lor a_2}
\subseteq D$ and from $\ext{p \lor q \lor r} \subseteq D \cup \ext{s}$
it follows that $\ext{a_1 \lor a_2} \subseteq D \cup \ext{s}$. It thus
only remains to be seen that $\ext{a_1 \lor a_2} \nsubseteq D$. Because
$\ext{p \lor q \lor r} \nsubseteq D$ there is some $u' \in \ext{p \lor q
\lor r}$ such that $u' \notin D$. Observe first that $u \leq u'$ because
$(-\infty,u) = C' \subseteq D$. By the choice of $p_r$ there is then a
$v' \in \ext{p_r}$ such that $u' \leq v'$. Clearly $v' \in \ext{a_1 \lor
a_2}$. We also have $v' \notin D$ because $D$ is convex, $u' \notin D$,
$u - 42 \in C' \subseteq D$ and $u - 42 < u \leq u' < v'$.
\end{proof}

\section{Conclusion}

We have shown that preferential conditional logic is complete with
respect to convexity over finite sets of points on the plane. Because of
the validities discussed in Section~\ref{line not enough} this result
can not be strengthened to convexity on the real line. There seem two be
two natural directions to continue this line of research. First, one
might ask what is the logic of finite sets of points on the line and
what is the logic of the real line. As our examples also show these
logics are not the same. Second, one might try to strengthen our
completeness result. Most interesting would be to show completeness with
respect to convexity over the complete plane, analogously to the
completeness of S4 with respect to the standard topology on the full
real line:
\begin{problem}
 Is preferential conditional logic complete with respect to models of
the form $(\R^2,\C,V)$, where $\C$ is the standard convexity 
and $V$ any valuation?
\end{problem}
It might be simpler to first show completeness with respect to bounded
regions in the plane. A plausible conjecture of this kind is the
following:
\begin{problem}
 Is preferential conditional logic complete with respect to models of
the form $(U,\C,V)$, where $U \subseteq \R^2$ is regular, compact and
convex, $\C$ is the relative convexity of $U$ in $\R^2$, $V$ is a
valuation that sends all propositional letters to regular closed sets,
and the propositional connectives are interpreted over the Boolean
algebra of regular closed sets?
\end{problem}
Note that by the Krein-Milman Theorem compact sets are in the closure of
their extreme points. Thus, one might hope that for the semantics of the
conditional they still behave similar to finite sets.

Another question for further research is how conditional logic relates
to other modal logics that have been developed to reason about convexity
or lines in space. Examples are the bimodal logics of lines and points
from \cite{Balbiani98,Venema99} or the logics of the one-step convexity
and betweenness modalities in \cite{Aiello02}. It seems that the
expressivity of the conditional is weak compared to the modalities in
these logics. Thus, one might hope to find interpretations of
preferential conditional logic into some of these more expressive
logics.

In this paper we have investigated the connections between conditional
logic and convexity from a purely formal perspective. It would be
interesting to see whether this new geometric semantics can lead to new
insights about applications such as the meaning counterfactual
conditionals in natural language or the structure of defeasible
reasoning.

\bibliographystyle{plain}
\bibliography{references}

\end{document}